\def\full{0}
\documentclass[11pt,english]{article}
\usepackage{babel,comment}
\usepackage{amsmath, amsthm, amssymb, url}
\usepackage{verbatim}
\usepackage{paralist}
\usepackage[usenames,dvipsnames]{color} 
\usepackage{multirow}
\usepackage{bigstrut}
\usepackage[disable]{todonotes}
\usepackage{fullpage, bbm}
\usepackage{tikz}
\usepackage{graphicx}
\usepackage{algorithm,algpseudocode}

\makeatletter
\def\BState{\State\hskip-\ALG@thistlm}
\makeatother

\ifnum\full=0
\usepackage[compact]{titlesec}
\titlespacing{\section}{0pt}{2ex}{1ex}
\titlespacing{\subsection}{0pt}{1ex}{1ex}
\titlespacing{\subsubsection}{0pt}{0.5ex}{0ex}

\renewenvironment{enumerate}[1]{\begin{compactenum}#1}{\end{compactenum}}
\fi

\newtheorem{theorem}{Theorem}[section]
\newtheorem{lemma}[theorem]{Lemma}

\newtheorem{claim}[theorem]{Claim}

\newtheorem{definition}{Definition}[section]

\numberwithin{figure}{section}

\newcommand{\eps}{{\epsilon}}

\newcommand{\ord}[2][th]{\ensuremath{{#2}^{\mathrm{#1}}}}

\newenvironment{myproof}{\begin{proof}
}{
\end{proof}}

\usepackage{fullpage}
\usepackage[numbers]{natbib}
\usepackage{enumerate}
\usepackage{enumitem}

\let\oldReturn\Return
\renewcommand{\Return}{\State\oldReturn}

\algdef{SE}[DOWHILE]{Do}{doWhile}{\algorithmicdo}[1]{\algorithmicwhile\ #1}%
\date{}

\title{Testing Unateness of Real-Valued Functions\footnote{This work was partially supported by NSF award CCF-1422975.}}
\author{Roksana Baleshzar\footnote{Pennsylvania State University, {\sf rxb5410@cse.psu.edu, mzm269@psu.edu, rxp271@cse.psu.edu, sofya@cse.psu.edu.}} \and Meiram Murzabulatov\footnotemark[2] \and Ramesh Krishnan S. Pallavoor\footnotemark[2] \and Sofya Raskhodnikova\footnotemark[2]}

\begin{document}

\maketitle
\begin{abstract}
We give a unateness tester for functions of the form $f:[n]^d\rightarrow R$, where $n,d\in \mathbb{N}$ and $R\subseteq \mathbb{R}$ with query complexity $O(\frac{d\log (\max(d,n))}{\epsilon})$. Previously known unateness testers work only for Boolean functions over the domain $\{0,1\}^d$. We show that every unateness tester for real-valued functions over hypergrid has query complexity $\Omega(\min\{d, |R|^2\})$. Consequently, our tester is nearly optimal for real-valued functions over $\{0,1\}^d$. We also prove that every nonadaptive, 1-sided error unateness tester for Boolean functions needs $\Omega(\sqrt{d}/\epsilon)$ queries. Previously, no lower bounds for testing unateness were known.
\end{abstract}

\section{Introduction}
We study property testing unateness of functions of the form $f:[n]^d\rightarrow R$, where $n,d\in \mathbb{N}$ and $R\subseteq \mathbb{R}$. Unate functions are used in mathematics and other technical disciplines, for example in tautology checking \cite{KC86, M14} and in switching theory \cite{S12}. For $x,y\in [n]^d$, where $x=x_1,\ldots,x_d$ and $y=y_1,\ldots,y_d$, we define $|x-y|_1=\sum\limits_{i=1}^n |x_i-y_i|$. Two points $x,y$ are {\sf neighbours} if $|x-y|_1=1$. A function $f:[n]^d\rightarrow \mathbb{R}$ is {\sf unate} if each dimension $i\in[d]$ can be assigned an associated direction {\sf up} or {\sf down}; the direction is {\sf up} if $f(x)\leq f(y)$ for all neighbours $x,y$ where $y_i=x_i+1$; and the direction is {\sf down} if $f(x)\geq f(y)$ for all such $x$ and $y$. Moreover, a function is {\sf monotone} if the direction is {\sf up} in all dimensions. Thus, unateness is a generalization of monotonicity.

The domain $[n]^d$ is called a {\sf hypergrid} and the special case $\{0,1\}^d$ is called a {\sf hypercube}. The {\sf distance} between two functions $f,g:[n]^d\rightarrow \mathbb{R}$ is equal to the fraction of points $x\in [n]^d$ where $f(x)\neq g(x)$. Given a parameter $\eps \in(0,1)$, two functions $f$ and $g$ are $\eps$-far from each other if the distance between $f$ and $g$ is at least $\eps$. A function $f$ is $\eps$-far from a property $P$ if it is $\eps$-far from any function which has property $P$. A {\sf property tester}~\cite{GGR96,RS96} for a property $P$ is a randomized algorithm which, given parameter $\eps \in(0,1)$ and oracle access to the input function $f$, accepts $f$ with probability $\frac{2}{3}$, if it has the property $P$, and rejects $f$ with probability $\frac{2}{3}$, if it is $\eps$-far from $P$. A tester is nonadaptive if it makes all queries in advance, and adaptive otherwise. A testing algorithm for property $P$ has $1$-sided error if it always accepts all inputs that satisfy $P$ and $2$-sided error, otherwise.
A {\sf unateness tester} is a property tester for unateness. 
To the best of our knowledge, testing unateness has been studied previously only for Boolean functions over the hypercube domain.


\subsection{Our results}
In this paper, we give a unateness tester for real-valued function over the hypergrid. It is the first unateness tester for this type of functions. We also present first lower bounds on the number of queries required to test unateness. In particular, our tester and lower bound differ only by a logarithmic factor for real-valued functions over the hypercube.

The query complexity of our unateness tester is $O(\frac{d\log (\max(d,n))}{\eps})$ for functions $f:[n]^d\rightarrow \mathbb{R}$. Our tester generalizes the tester of~\citet{KS16} to real-valued functions over the hypergrid domain. Their tester has query complexity $O(\frac{d\log d}{\eps})$ and it works for Boolean functions over the hypercube domain.

We also present two lower bounds for testing unateness of functions over the hypercube domain. The first lower bound of $\Omega(\min\{d, |R|^2\})$, where $R$ is the range of the function, is over the real-valued functions. Our second lower bound of $\Omega(\sqrt{d}/\epsilon)$ is for nonadaptive, 1-sided error unateness testers over Boolean functions. Our lower bounds build on the same lower bounds for monotonicity testing by \citet{BBM12} and \citet{FLNRRS02} of Boolean functions over the hypercube domain.

\subsection{Previous work}
The problem of testing unateness was first considered by \citet{GGLRS00}. In their work, by extending their monotonicity tester, they obtain a nonadaptive, $1$-sided error tester for unateness with query complexity $O(\frac{d^{3/2}}{\epsilon})$. Recently, ~\citet{KS16} gave an adaptive unateness tester with query complexity $O(\frac{d\log d}{\eps})$. 

The related property of monotonicity has been studied extensively for different types of functions in the context of property testing. The problem of testing monotonicity of Boolean functions over the hypercube domain was first introduced by \citet{GGLRS00}. It was shown in \cite{DGLRRS99,GGLRS00} that monotonicity for hypercube domain can be tested with query complexity $O(\frac{d}{\epsilon})$. A monotonicity tester with better query complexity of $\tilde{O}(\frac{d^{7/8}}{\epsilon^{3/2}})$ was introduced by \citet{CS13b}. \citet{CST14} modified the tester of \citet{CS13b} and improved the query complexity to $\tilde{O}(\frac{d^{5/6}}{\epsilon^4})$. Most recently, \citet{KMS15} improved the query complexity of the tester to $\tilde{O}(\frac{\sqrt{d}}{\epsilon^2})$.

The lower bound for any nonadaptive one-sided error tester for monotonicity over the hypercube was proved to be $\Omega(\sqrt{d})$ by \citet{FLNRRS02}. Also, \citet{CDST15} gave a lower bound of almost $\Omega(\sqrt{d})$ for any nonadaptive, two-sided error tester. Moreover, there is a lower bound of $\Omega(\min\{d, |R|^2\})$ over the real-valued functions by \citet{BBM12}. Most recently, \citet{BB16} gave a lower bound of $\tilde{\Omega}(d^{\frac{1}{4}})$ for adaptive testers of Boolean functions over the hypercube. \citet{CS13} proved that any adaptive, two-sided monotonicity tester for functions $f:[n]^d\rightarrow \mathbb{N}$ must make $\Omega(\frac{d\log n- \log \eps^{-1}}{\eps})$ queries.

\section{Algorithm for Testing Unateness}\label{sec:upper-bound-for-unateness}
In this section, we present an algorithm to test unateness of real-valued functions over the hypergrid domain and prove Theorem~\ref{thm:thm41}.

We use the notation $z\circ_T w$ to denote concatenation with respect to the locations in the set $T$.
More specifically, this operation fills the locations in $T$ with $|T|$ bits in $z$, and the other $[d]\setminus T$ locations with $d-|T|$ bits in $w$, respectively.

\begin{theorem}\label{thm:thm41}
For a given function $f:[n]^d\rightarrow \mathbb{R}$, there exists a $1$-sided error unateness tester with a query complexity of $O(\frac{d\log (max(d,n))}{\eps})$, where $\eps \in (0,\frac{1}{2})$ is the proximity parameter.
\end{theorem}

\begin{proof}
The main idea of the tester described in Algorithm~\ref{Unateness_Tester} is as follows. Given a function $f:[n]^d\rightarrow \mathbb{R}$, the tester finds a subset of {\sf influential} dimensions $T\subseteq [d]$, such that the dependence of $f$ on the remaining dimensions, $[d]\setminus T$, is negligible. Furthermore, for each influential dimension $i\in T$, the tester finds an edge $(x,x+e_i)$, where $f(x)\neq f(x+e_i)$, and assigns a 
{\sf direction} to the $i$-th dimension according to the value of $f$ at these two points. The direction is {\sf up} if  $f(x)<f(x+e_i)$ and it is {\sf down} if  $f(x)> f(x+e_i)$. After that, the algorithm applies a monotonicity tester on $f$ with respect to the directions for the dimensions in $T$ and some random assignments $w$ for other dimensions, and outputs the answer.

\begin{algorithm}
\caption{{\sf Unateness Tester}$(f:[n]^d\rightarrow \mathbb{R},\eps)$}
\label{Unateness_Tester} 
\begin{algorithmic}[1]
\State Let $m=O(\frac{d}{\eps})$ and $T=\emptyset$
\For {$i=1...m$}
\State {\sf Find\_an\_Influential\_Dimension}$(f,T)$
\If {it returns a dimension and its direction $(i^*,b_{i^*})$}
\State Add $i^*$ to $T$, and let $b_{i^*}$ be its direction
\EndIf
\EndFor
\State Pick a uniform $w\in[n]^{[d]\setminus T}$
\State Let $f_w(z)=f(z\circ_T w)$ for all $z\in [n]^T$
\State Apply the monotonicity tester by \citet{CS13} on $f_w$ with respect to the directions $\{b_i:i\in T\}$ and proximity parameter $\frac{\eps}{2}$
\Return the output of the monotonicity tester
\end{algorithmic}
\end{algorithm}

In order to find influential dimensions, the tester uses the function {\sf Find\_an\_Influential\_Dimension}, which is presented in Algorithm~\ref{Find_an_Influential_Dimension}. It works as a subroutine and tries to find an influential dimension in each run. It gets oracle access to a function $f:[n]^d\rightarrow \mathbb{R}$, and a subset of dimensions $T\subseteq [d]$, which is empty at the beginning. This subroutine outputs $\perp$ if it doesn't find any new influential dimension. Otherwise, it outputs a dimension $i^*\in [d]\setminus T$ and its direction $b_{i^*}\in\{{\sf up,down}\}$.

This subroutine can find a new influential dimension and its direction with high probability if $f$ has non-negligible dependence on the dimensions in $[d]\setminus T$. Each time we run {\sf Find\_an\_Influential\_Dimension}, it picks two points $x,y\in[n]^d$ uniformly and independently at random, whose coordinates are equal in the dimensions in $T$. If $f(x)\neq f(y)$, it uses {\sf binary search} to decrease the distance between $x$ and $y$ to $1$, while ensuring that $f(x)\neq f(y)$ still holds.

\begin{algorithm}
\caption{{\sf Find\_an\_Influential\_Dimension}$(f:[n]^d\rightarrow \mathbb{R},T)$}
\label{Find_an_Influential_Dimension} 
\begin{algorithmic}[1]
\State Pick $x,y \in_R [n]^d$ independently and uniformly at random such that $x_T=y_T$
\If{$f(x)=f(y)$}
\Return $\perp$
\Else
\Do
\State $U\leftarrow \{i\in [d]:x_i=y_i\}$
\State $V\leftarrow \{i\in [d]:x_i\neq y_i\}$
\State Let $V = \{v_1,...,v_t\}$
\State Set $z_V=x_{v_1}\circ ...\circ x_{v_{\lfloor\frac{t}{2}\rfloor}}\circ y_{v_{\lfloor\frac{t}{2}\rfloor+1}}\circ ...\circ y_{v_t}$
\State Let $z=x_U\circ_T z_V \in [n]^d$
\If {$f(x)\neq f(z)$}
\State $y\leftarrow z$
\Else
\State $x\leftarrow z$
\EndIf
\doWhile {$|V|\neq1$}
\State Let $i^*\in[d]$ be the element in $V$
\State Let $b\in\{up,down\}$ be the direction of edge $(x,y)$
\Return $(i^*,b)$
\EndIf
\end{algorithmic}
\end{algorithm}

If the input function $f$ is unate, then the tester always accepts, as Algorithm~\ref{Find_an_Influential_Dimension} always finds the correct direction for each influential dimension and the function with respect to the influential dimensions is monotone.

In order to show the correctness of Algorithm~\ref{Unateness_Tester}, we will prove the following two lemmas. In these lemmas let $r$ denote the number of different values in the image of the function $f$, and $P_{k}^{z}$ where $k\in[r]$ and $z\in [n]^T$ be the fraction of the $k^{\text{th}}$ frequent value among the values of $f(x|_{x_T=z})$.

\begin{lemma}\label{lem:lem41}
Let $f:[n]^d\rightarrow \mathbb{R}$ and $T\subseteq[d]$ be the set in Algorithm~\ref{Unateness_Tester} after running $m=O(\frac{d}{\eps})$ iterations of Algorithm~\ref{Find_an_Influential_Dimension}. Then, with probability at least $5/6$, the set $T$ satisfies
$$\mathbb{E}_{z\in [n]^{T}}\left[\sum\limits_{k=1}^{r} P_{k}^{z}(1-P_{k}^{z}) \right]< \frac{\eps}{16}.$$
\end{lemma}

\begin{lemma}\label{lem:lem43}
Let $f:[n]^d\rightarrow \mathbb{R}$ be a real-valued function, and let $T\subseteq[d]$ such that,
\begin{equation}\label{eq:eq1}
\mathbb{E}_{z\in [n]^{T}}\left[\sum\limits_{k=1}^{r} P_{k}^{z}(1-P_{k}^{z}) \right]< \frac{\eps}{16}.
\end{equation}
Then, for a random $w\in[n]^{[d]\setminus T}$,
$$\Pr_{w\in[n]^{[d]\setminus T}}[{\tt dist}(f_w,f)>\frac{\eps}{2}]\leq\frac{1}{6}.$$
\end{lemma}

\begin{proof}[Proof of Lemma~\ref{lem:lem41}]
First, we show that $\mathbb{E}_{z\in [n]^{T}}\left[\sum\limits_{k=1}^{r} P_{k}^{z}(1-P_{k}^{z})\right]$ is equal to the probability that Algorithm~\ref{Find_an_Influential_Dimension} finds a new dimension on input $(f,T)$. Algorithm~\ref{Find_an_Influential_Dimension} finds a new dimension if and only if $f(x)\neq f(y)$ in Line $2$. By definition of $P_{k}^{z}$, the probability of choosing $x$ such that $f(x)=k$ is $P_{k}^{z}$, and the probability of choosing $y$ such that $f(y)\neq k$ is $1-P_{k}^{z}$. The total probability of choosing two different values is $\sum\limits_{k=1}^{r} P_{k}^{z}(1-P_{k}^{z})$ for a fixed $z$. So, the probability over all possible $z$'s is $\mathbb{E}_{z\in [n]^{T}}\left[\sum\limits_{i=1}^{r} P_{k}^{z}(1-P_{k}^{z})\right]$.

Now, we state the following claim to prove that when the Algorithm~\ref{Find_an_Influential_Dimension} finds a new dimension then the value of $\mathbb{E}_{z\in [n]^{T}}\left[\sum\limits_{k=1}^{r} P_{k}^{z}(1-P_{k}^{z})\right]$ does not increase.

\begin{claim}\label{clm:clm42}
For $\forall T\subseteq [d]$ and $\forall i\in [d]$, when the dimension $i$ is added to the set $T$, the value of
$$\mathbb{E}_{z\in [n]^{T}}\left[\sum\limits_{k=1}^{r} P_{k}^{z}(1-P_{k}^{z})\right]$$
does not increase.
\end{claim}

\begin{proof}
Let $P_{k,T}^{z}$ denote $P_{k}^{z}$ when $z\in[n]^{T}$. Let $T'=T \cup \{i\}$. Define $z_j$, to be equal to $z$ for all dimensions in $T$, and equal to $j$ in the $i^{\text{th}}$ dimension, for each $j \in [n]$. To simplify notation for the rest of the proof, let $P=P_{k,T}^{z}$, $P^j=P_{k,T'}^{z_j}$, for $j \in [n]$. Then, $P={\sum\limits_{j=1}^{n} \frac{P^j}{n}}$. The probability that Algorithm~\ref{Find_an_Influential_Dimension} finds a new dimension on input $(f,T')$ is
$$\mathbb{E}_{z\in [n]^{T}} \left[\sum\limits_{k=1}^{r} \sum\limits_{j=1}^{n} \frac{P^j(1-P^j)}{n} \right].$$
To complete the proof, it is enough to show that $g(P)\geq \sum\limits_{j=1}^{n}\frac{g(P^j)}{n}$, where $g(x)=x(1-x)$. This inequality holds because $P={\sum\limits_{j=1}^{n} \frac{P^j}{n}}$ and $g$ is concave.
\end{proof}

The following claim completes the proof of Lemma~\ref{lem:lem41}.
\begin{claim}\label{clm:clm43}
After $m=O(\frac{d}{\eps})$ iterations of Algorithm~\ref{Find_an_Influential_Dimension}, either the size of the set $T$ is equal to $d$, or with probability at least $5/6$, the following inequality holds,
$$\mathbb{E}_{z\in [n]^{T}} \left[\sum\limits_{k=1}^{r} P_{k}^{z}(1-P_{k}^{z}) \right]< \frac{\eps}{16}.$$
\end{claim}

\begin{proof}
Let $Y_i$ be an indicator random variable, which is equal to $1$ if Algorithm~\ref{Find_an_Influential_Dimension} finds a new dimension in the $i^{\text{th}}$ iteration and $0$, otherwise. We have $\Pr[Y_i]=\mathbb{E}_{z\in [n]^{T_i}}[Y_i]=\mathbb{E}_{z\in [n]^{T_i}}\left[\sum\limits_{k=1}^{r} P_{k}^{z}(1-P_{k}^{z})\right]$, where $T_i$ is the set of influential dimensions in the $i^{\text{th}}$ iteration. Hence, the total number of dimensions that Algorithm~\ref{Find_an_Influential_Dimension} finds after $m$ iterations is $Y=\sum\limits_{i=1}^{m} Y_i$.
Now, after $m$ iterations, if the following holds:
$$\mathbb{E}_{z\in [n]^{T}}\left[\sum\limits_{k=1}^{r} P_{k}^{z}(1-P_{k}^{z}) \right]\geq \frac{\eps}{16},$$
then,
$$\mathbb{E}_{z\in [n]^{|T|}}[Y]\geq \frac{m \eps}{16}.$$
Using Chernoff Bound we have,
$$\Pr_{z\in [n]^{T}}[Y \leq \frac{m\eps}{32}] \leq e^{-\frac{m\eps}{128}}.$$
For $m=\frac{32d}{\eps}$ and $d\geq 8$, we have
$$\Pr_{z\in [n]^{T}}[Y\leq d]\leq e^{-\frac{d}{4}} \leq \frac{1}{6}.$$
\end{proof}
This completes the proof of Lemma~\ref{lem:lem41}.
\end{proof}

\begin{proof}[Proof of Lemma~\ref{lem:lem43}]
Define the plurality function ${\tt Pl}_T:[n]^T\rightarrow [r]$ as ${\tt Pl}_T(z)=k$ such that $P_{k}^{z}$ is the maximum (ties are broken arbitrarily).

By assumption~(\ref{eq:eq1}) for a uniformly random $w\in[n]^{[d]\setminus T}$, it holds that
\begin{align*}
    \mathbb{E}_{w\in[n]^{[d]\setminus T}}[{\tt dist}(f_w,{\tt Pl}_T)] &= \mathbb{E}_{w\in[n]^{[d]\setminus T}}\left[\Pr_{z\in[n]^{T}}[f_w(z)\neq {\tt Pl}_T(z)]\right] \\
    &= \mathbb{E}_{w\in[n]^{[d]\setminus T}} \left[\mathbb{E}_{z\in[n]^{T}}[\mathbbm{1}(f_w(z)\neq {\tt Pl}_T(z))]\right] \\
    &= \mathbb{E}_{z\in[n]^{T}} \left[\mathbb{E}_{w\in{[n]^{[d]\setminus T}}}[\mathbbm{1}(f_w(z)\neq {\tt Pl}_T(z))]\right] \\
    &= \mathbb{E}_{z\in[n]^{T}}\left[\Pr_{w \in [n]^{[d]\setminus T}}[f_w(z)\neq {\tt Pl}_T(z)]\right].
\end{align*}

Without loss of generality, assume that $P_{1}^{z}$ is the probability of the value with the highest frequency in $f(x|_{x_T=z})$. Hence,
$$\mathbb{E}_{z\in[n]^{T}}\left[\Pr_{w\in[n]^{[d]\setminus T}}[f_w(z)\neq {\tt Pl}_T(z)]\right] = \mathbb{E}_{z\in[n]^{T}}[1-P_{1}^{z}].$$

\begin{claim}\label{clm:clm45}
For all $z\in [n]^{T}$,
$$1-P_{1}^{z} \leq \sum\limits_{k=1}^{r} P_{k}^{z}(1-P_{k}^{z}).$$
\end{claim}

\begin{proof}
To simplify notations for the rest of the proof, let $P_k=P_{k}^{z}$. We have
$$\sum\limits_{k=1}^{r} P_{k}(1-P_{k}) = 1-\sum\limits_{k=1}^{r} P_{k}^{2}.$$
So it is enough to show that, $1-P_{1} \leq 1-\sum\limits_{k=1}^{r} P_{k}^{2}$ or $\sum\limits_{k=1}^{r} P_{k}^{2} \leq P_1 $ which can be proved as follows,
\begin{align*}
\sum\limits_{k=1}^{r} \frac{P_{k}^{2}}{P_1} \leq \sum\limits_{k=1}^{r} \frac{P_{k}^{2}}{P_k} = 1
\end{align*}
\end{proof}

By Claim~\ref{clm:clm45} and our initial assumption~(\ref{eq:eq1}), we have
$$\mathbb{E}_{z\in[n]^{T}}[1-P_{1}^{z}] \leq \mathbb{E}_{z\in[n]^{T}}\left[\sum\limits_{k=1}^{r} P_{k}^{z}(1-P_{k}^{z})\right] < \frac{\eps}{16}.$$
So,
$$\mathbb{E}_{w\in[n]^{[d]\setminus T}}[{\tt dist}(f_w,{\tt Pl}_T)]<\frac{\eps}{16}.$$
By Markov's inequality,
$$\Pr_{w\in[n]^{[d]\setminus T}}\left[{\tt dist}(f_w,{\tt Pl}_T)\geq \frac{3\eps}{8}\right]\leq \frac{1}{6}.$$
Moreover,
$${\tt dist}(f,{\tt Pl}_T)=\mathbb{E}_{w\in[n]^{[d]\setminus T}}[{\tt dist}(f_w,{\tt Pl}_T)]<\frac{\eps}{16}.$$
Therefore, by triangle inequality we have,
$$\Pr_{w\in[n]^{[d]\setminus T}}\left[{\tt dist}(f(x|_{x_{[d]\setminus T}=w}),f) \geq \frac{3 \eps}{8}+\frac{\eps}{16}] \leq \Pr_{w\in[n]^{[d]\setminus T}}[{\tt dist}(f_w,{\tt Pl}_T)\geq \frac{7\eps}{16}\right] \leq \frac{1}{6}.$$
This completes the proof of Lemma~\ref{lem:lem43}.
\end{proof}

The proof of Theorem~\ref{thm:thm41} follows from Lemmas~\ref{lem:lem41} and \ref{lem:lem43}. Let $T\subseteq [d]$ be the set of influential dimensions in the unateness tester after $m=O(\frac{d}{\eps})$ iterations of Algorithm~\ref{Find_an_Influential_Dimension}. By Lemma~\ref{lem:lem41} the following inequality holds with probability at least $5/6$
$$\mathbb{E}_{z\in [n]^{T}}\left[\sum\limits_{k=1}^{r} P_{k}^{z}(1-P_{k}^{z})\right]< \frac{\eps}{16}.$$

If $T$ satisfies the above inequality and $f$ is $\eps$-far from being unate, then by Lemma~\ref{lem:lem43}, for a random $w\in[n]^{[d]\setminus T}$, the function $f(x|_{x_{[d]\setminus T}=w})$ is $\frac{\eps}{2}$-far from being unate with probability at least $\frac{5}{6}$. And with the same probability it is $\frac{\eps}{2}$-far from being monotone with respect to the directions of dimensions in the set $T$. So, with high probability the monotonicity tester in Algorithm~\ref{Unateness_Tester} will reject $f(x|_{x_{[d]\setminus T}=w})$ with respect to the directions of the dimensions in $T$.

Now, we analyze the query complexity of Algorithm~\ref{Unateness_Tester}. Each call of Algorithm~\ref{Find_an_Influential_Dimension} makes $O(\log{d})$ queries to perform the binary search. Algorithm~\ref{Find_an_Influential_Dimension} is called $O(\frac{d}{\eps})$ times by Algorithm~\ref{Unateness_Tester}. At the end, Algorithm~\ref{Unateness_Tester} runs monotonicity tester with query complexity $O(\frac{d\log{n}}{\eps})$. So, the total query complexity of Algorithm~\ref{Unateness_Tester} is $O(\frac{d\log (max(d,n))}{\eps})$.

This completes the proof of Theorem~\ref{thm:thm41}.
\end{proof}

\section{Lower Bounds for Unateness}\label{sec:lower-bound-for-unateness}
\subsection{Lower Bound for Real-Valued Functions}\label{sec:lower-bound-for-real-valued-functions}
In this section, we prove an adaptive, 2-sided error lower bound for testing unateness of functions over the hypergrid domain. This lower bound technique builds on the one used for proving the lower bound for monotonicity by \citet{BBM12}.

\begin{theorem}\label{thm:thm1}
Testing $f:\{0,1\}^d \rightarrow R$ for unateness requires $\Omega (min\{d, |R|^2\})$ queries.
\end{theorem}

\begin{proof}
First, we consider the case when $R=\mathbb{R}$. We will prove the lower bound for testing unateness in this case by a reducton from {\sf Set-Disjointness} problem.

\begin{definition}{(Set-Disjointness)}
In the {\sf Set-Disjointness} problem, there are two players Alice and Bob. Alice receives a string $S$ as input and Bob receives a string $T$ as input, where both strings have size $k$. Alice and Bob also have access to a random string. They must compute the value of function $DISJ_k(S,T)$, which is $1$ if $S\cap T=\emptyset$ and $0$ otherwise.
\end{definition}

\begin{theorem}\label{thm:thm2}
Testing $f:\{0,1\}^d \rightarrow \mathbb{R}$ for unateness requires $\Omega (d)$ queries.
\end{theorem}
\begin{proof}
Let $S,T\subseteq [d]$ be the subsets received by Alice and Bob as the input to an instance of the {\sf Set-Disjointness} problem. Alice and Bob build functions $\chi_S, \chi_T:\{0,1\}^d \rightarrow \{-1,1\}$, respectively, where $\chi_S(x) = (-1)^{\sum\nolimits_{i \in S} x_i}$ and $\chi_T(x) = (-1)^{\sum\nolimits_{i \in T} x_i}$. Let $h:\{0,1\}^d \rightarrow \mathbb{Z}$ be a function defined by $h(x) = 2|x|+\chi_S(x)+\chi_T(x)$.

\begin{claim}\label{clm:clm1} The following holds:
\begin{enumerate}[label=(\alph*)]
    \item The function $h$ is unate if $S$ and $T$ are disjoint.
    \item The function $h$ is $\frac{1}{8}$-far from unate, otherwise.
\end{enumerate}
\end{claim}

\begin{proof}
Fix $i \in {[d]}$. For every $x\in \{0,1\}^d$, we define $x^0 , x^1\in \{0,1\}^d$ to be the vectors obtained by fixing the $i^\text{th}$ coordinate of $x$ to $0$ and $1$, respectively.\\
Fist, we prove part $(a)$. Since $S$ and $T$ are disjoint, there are three cases:
\begin{enumerate}
    \item $i\notin S$ and $i\notin T$. Then, $$h(x^1)-h(x^0)=2|x^1|-2|x^0|=2>0.$$
    \item $i\in S$ and $i\notin T$. Then, $$h(x^1)-h(x^0)=2|x^1|-2|x^0|-2\chi_S(x^0)\geq 0.$$
    \item $i\notin S$ and $i\in T$. Then, $$h(x^1)-h(x^0)=2|x^1|-2|x^0|-2\chi_T(x^0)\geq 0.$$
\end{enumerate}
So, when $S$ and $T$ are disjoint, the function $h$ is non-decreasing along dimension $i$.

For the part $(b)$ of our claim, since $S\cap T \neq \phi$ there exists $i \in S \cap T$. For this $i$, we will show that there exists $\frac{1}{4}$-fraction of pairs of the form $(x^0,x^1)$ such that $h(x^1) < h(x^0)$ and another $\frac{1}{4}$-fraction of pairs such that $h(x^1) > h(x^0)$ implying that $h$ is not unate. We have
$$h(x^1)-h(x^0)=2|x^1|-2|x^0|-2\chi_S(x^0)-2\chi_T(x^0).$$
It is clear that if $\chi_S(x^0) = \chi_T(x^0) = 1$ then $h(x^1) < h(x^0)$. And if $\chi_S(x^0) = \chi_T(x^0) = -1$ then $h(x^1) > h(x^0)$. Suppose that $j\in S$ and also $j\neq i$. Let $\bar{x}^0$ be $x^0$ with the $j^\text{th}$ bit flipped. Then, 
$\chi_S(x^0) = -\chi_S(\bar{x}^0)$
, which means $\chi_S(x^0)=1$ with probability $\frac{1}{2}$. By the same argument, $\chi_T(x^0)=1$ with probability $\frac{1}{2}$. So, at least for $\frac{1}{4}$ of these pairs we have $\chi_S(x^0) = \chi_T(x^0) = 1$. By using the same technique we can prove that for at least $\frac{1}{4}$ fraction of these pairs we have $\chi_S(x^0) = \chi_T(x^0) = -1$. 
Therefore, when $S$ and $T$ are overlapping, $h$ is $\frac{1}{8}$-far from unate.
\end{proof}
This completes the proof of Theorem~\ref{thm:thm2}
\end{proof}

Now consider the case when $|R|>12\sqrt{d}+5$. We can again use a reduction from {\sf set-disjointness} problem.
\begin{theorem}\label{thm:thm3}
Testing $f:\{0,1\}^d \rightarrow R$ for unateness, where $|R|>12\sqrt{d}+5$, requires $\Omega (d)$ queries.
\end{theorem}
\begin{proof}
Without lost of generality we can also assume that $\{d-6\sqrt{d}-2,...,d+6\sqrt{d}+2\}\subseteq R$. Let $\chi_S,\chi_T:\{0,1\}^d\rightarrow \{-1,1\}$ be the functions given to Alice and Bob in the communication game, where $S,T\subseteq [d]$. Now we will define the function $h':\{0,1\}^d\rightarrow \mathbb{Z}$ (as defined in \cite{BBM12}), as follows:

\begin{equation*}
h'(x) =
\begin{cases}
  d-6\sqrt{d}-2 & \text{if $|x| < \frac{d}{2}-3\sqrt{d}$;} \\
  d+6\sqrt{d}+2 & \text{if $|x| > \frac{d}{2}+3\sqrt{d}$;} \\
  h(x) & \text{otherwise.}
\end{cases}
\end{equation*}

This function rounds $h(x)$ up to $d-6\sqrt{d}-2$ for strings with low Hamming distance and rounds $h(x)$ down to $d+6\sqrt{d}+2$ for strings with high Hamming distance. It also keeps the actual value of $h$ for points $x$ such that $\frac{d}{2}-3\sqrt{d}\leq|x|\leq \frac{d}{2}+3\sqrt{d}$. By the definition of $h(x)$, this value will be between $d-6\sqrt{d}-2$ and $d+6\sqrt{d}+2$.

\begin{claim}\label{clm:clm2}
Function $h'$ is unate if $h$ is unate and $h'$ is $1/16$-far from unate when $h$ is $1/8$-far from unate.
\end{claim}
\begin{proof}
Since every monotone function is also unate, the first part of the claim directly follows from \citet{BBM12}.

For the second part we prove that the distance between $h$ and $h'$ is at most $\frac{1}{16}$. Thus, when $S$ and $T$ intersect, $h$ is $\frac{1}{8}$-far from unate and so $h'$ is $\frac{1}{16}$-far from unate. 
Let $x$ be a uniformly random point from $\{0,1\}^d$. We know that $h'(x)=h(x)$ when $\frac{d}{2}-3\sqrt{d}\leq|x|\leq \frac{d}{2}+3\sqrt{d}$. By Chernoff Bound,
$$\Pr\left[ \left| |x|-\frac{d}{2} \right| > 3\sqrt{d} \right] < 0.03 < \frac{1}{16} \Rightarrow \Pr[h'(x)\neq h(x)]<\frac{1}{16}.$$
By triangle inequality, $h'$ is at least $\frac{1}{16}$-far from unate.
\end{proof}
This completes the proof of Theorem~\ref{thm:thm3}.
\end{proof}

For the last case, we prove the following theorem.
\begin{theorem}\label{thm:thm4}
Testing a function $f:\{0,1\}^d \rightarrow R$ for unateness, where  $|R|=o(\sqrt{d})$, requires $\Omega (|R|^2)$ queries.
\end{theorem}
\begin{proof}
In this case, we can prove the lower bound by using a reduction to Theorem~\ref{thm:thm3}. Let $m$ be the largest integer such that $|R| > 12\sqrt{m}+5$. 
We prove the following claim.
\begin{claim}\label{clm:clm3}
For every function $g:\{0,1\}^m\rightarrow R$, there exists a function $h:\{0,1\}^d\rightarrow R$ satisfying the following conditions:

\begin{enumerate}
    \item If $g$ is unate, then $h$ is also unate.
    \item If $g$ is $\eps$-far from unate, then $h$ is also $\eps$-far from unate.
    \item For all $x\in\{0,1\}^d$, we can find $h(x)$ with one query to $g$.
\end{enumerate}
\end{claim}

\begin{proof}
We use padding technique to construct $h$ from $g$. Let $h(x,y)=g(x)$ where $x\in\{0,1\}^m$ and $y\in\{0,1\}^{d-m}$. The first and third conditions can be directly deduced from the definition of $h$. For the second condition, suppose, for the sake of contradiction that $h$ is $\eps$-close from unate. Then, there exists a unate function $\tilde{h}$, such that $\Pr_{x,y}[\tilde{h}(x,y)\neq h(x,y)]<\eps$. By an averaging argument, we can find a $\tilde{y}\in\{0,1\}^{d-m}$ where, $\Pr_{x}[\tilde{h}(x,\tilde{y})\neq h(x,\tilde{y})]<\eps$. Now, let $\tilde{g}(x)=\tilde{h}(x,\tilde{y})$ for all $x\in\{0,1\}^m$. Thus, $\tilde{g}$ is unate and $\Pr_{x}[\tilde{g}(x)\neq g(x)]<\eps$, contradicting the assumption that $g$ is $\eps$-far from unate.
\end{proof}

Now we can construct a tester for a given input $g:\{0,1\}^m\rightarrow R$. Suppose, $h$ is the function that satisfies Claim \ref{clm:clm3}. We run the unateness tester on $h$ and output its answer. By Claim~\ref{clm:clm3}, correctness of this tester follows from the correctness of the tester for $h$. Since in Theorem \ref{thm:thm3} we proved that the unateness tester for $g$ uses $\Omega(m)$ queries and Claim \ref{clm:clm3} condition $3$ holds, the tester for $h$ uses the same number of queries. This means testing a function $f:\{0,1\}^d \rightarrow R$ for unateness, where $|R|=o(\sqrt{d})$, requires $\Omega(m)=\Omega(|R|^2)$ queries.
\end{proof}

Theorem \ref{thm:thm1} follows from Theorems~\ref{thm:thm2}, \ref{thm:thm3} and \ref{thm:thm4}.
\end{proof}

\subsection{Lower Bound for Boolean Functions}\label{sec:lower-bound-for-boolean-functions}
In this section, we prove a nonadaptive, 1-sided error lower bound for testing unateness of Boolean functions over the hypercube domain. This lower bound technique is similar to the one used for proving the lower bound for monotonicity by \citet{FLNRRS02}.

\begin{theorem}\label{thm:one-sided-non-adaptive-lower-bound-unateness}
Every nonadaptive 1-sided error tester for unateness of functions of the form $f:\{0,1\}^d \rightarrow \{0,1\}$ must make $\Omega(\sqrt{d})$ queries.
\end{theorem}

The proof of the theorem is by the application of Yao's minimax principle. A 1-sided error tester must reject only if it finds a violation, otherwise it must accept. We define a distribution over a set of ``hard" functions and prove that for this hard distribution, a tester with query complexity $q$ detects a violation with probability at most $O(q/\sqrt{d})$.

For $x\in\{0,1\}^d$, define $|x|$ as the number of bits that are equal to $1$. The set of hard functions is defined as follows: for every $i\in\left[\frac{d}{2}\right]$, define $f_i:\{0,1\}^d \rightarrow \{0,1\}$ as
\begin{equation*}
f_i(x_1,\ldots,x_d) =
\begin{cases}
  1 & \text{if $|x| > d/2 + \sqrt{d}$} \\
  0 & \text{if $|x| < d/2 - \sqrt{d}$} \\
  x_i \oplus x_{i+\frac{d}{2}}& \text{otherwise}
\end{cases}
\end{equation*}

\begin{claim} \label{claim:claim1} 
For every $i\in\left[\frac{d}{2}\right]$, $f_i$ is $\epsilon$-far from unate for some constant $\epsilon$ such that $0 < \epsilon \leq \frac{1}{4}$.
\end{claim}

\begin{proof}
For a fixed $i \in \left[\frac{d}{2}\right]$, group together all points that differ only in the $\ord{i}$ and $\ord{\left(i+\frac{d}{2}\right)}$ coordinate. There are $4$ points in each group as illustrated in Figure~\ref{fig:fig1}. Overall, there are $2^{d-2}$ such independent groups.

It can be seen that a constant fraction of such groups violate unateness of $f_i$ (by its definition). In each such group, the value of $f_i$ on at least $1$ point has to be changed (or ``repaired") to make it unate. Thus, overall, a constant fraction of points have to be changed to make the function unate. Hence, $f_i$ is $\epsilon$-far from unate for some constant $\epsilon$.


\begin{figure}[ht]
\centering
\includegraphics[scale=0.5]{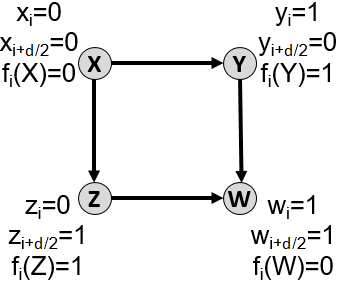}
\caption{Illustration of ``hard" function}
\label{fig:fig1}
\end{figure}
\end{proof}

\begin{proof}[Proof of Theorem~\ref{thm:one-sided-non-adaptive-lower-bound-unateness}]
Let $Q$ be the queried set of points from $\{0,1\}^d$ of size $q$. By the definition of 1-sided error tester, we need to sample at least one violated, comparable pair of points to reject the function (A pair of points $(x,y)$ is \emph{comparable} if either $x \prec y$ or $y \prec x$). So, the tester can detect a violated pair for $f_i$ if there exists two queries $u,v$ such that they are comparable and differ in either $\ord{i}$ or $\ord{\left(i+\frac{d}{2}\right)}$ coordinate.
Construct an undirected graph with vertex set $Q$ and edge set $\{(x,y)| x \text{ and } y \textnormal{ are comparable}, x,y \in Q\}$. Consider a spanning forest of this graph. If there exists a violation by function $f_i$ in pair $(u,v)$ they should be in the same tree. Moreover, there should exist an adjacent pair of vertices in the path between $u$ and $v$ which differ in the $\ord{i}$ bit. Therefore, the total number of functions $f_i$ for which the query set could find the violation is at most the maximum number of edges in the spanning forest (which is at most $q-1$) times the maximum possible distance between two adjacent vertices (which is at most $2\sqrt{d}$ with respect to our function because all violations are for $x \in \{0,1\}^d$ such that $d/2 - \sqrt{d} \leq |x| \leq d/2 + \sqrt{d}$). So at most $O(q\sqrt{d})$ functions can be detected for violation. We have $\frac{d}{2}$ different functions. Hence, the number of queries required is at least $\Omega(\sqrt{d})$.
\end{proof}

\section{Open Questions}\label{sec:open-questions}
Our bounds for testing unateness are nearly tight for real-valued functions. However, there is a gap between the upper and lower bounds for Boolean functions. The first open question is to close this gap. In this paper, we consider testing unateness of real-valued functions with respect to the Hamming distance. Can one design testers with respect to other metrics such as $L_p$-distances, as defined in \cite{BRY14}?

\bibliographystyle{plainnat}
\bibliography{references}

\end{document}